\documentclass[11pt]{article}

\usepackage{PRIMEarxiv}

\usepackage{amsmath,amssymb,amsthm}
\usepackage{hyperref}
\usepackage{tikz}
\usepackage{xcolor}
\usepackage{orcidlink}
\usepackage{booktabs}
\usetikzlibrary{arrows.meta,positioning}

\definecolor{royalblue}{RGB}{65,105,225}
\newcommand{\michael}[1]{{\color{black}#1}}
\newcommand{\andre}[1]{\color{black}#1}
\title{Selecting a Maximum Solow-Polasky Diversity Subset in General Metric Spaces Is NP-hard}

\author{
Michael T.\ M.\ Emmerich\,\orcidlink{0000-0002-7342-2090}\\
Faculty of Information Technology, University of Jyv\"askyl\"a, Finland
\And
Ksenia Pereverdieva\,\orcidlink{0000-0002-9111-8359}\\
Leiden Center of Advanced Computer Science, Leiden University, The Netherlands
\And
Andr\'e Deutz\,\orcidlink{0000-0002-9047-6533}\\
Leiden Center of Advanced Computer Science, Leiden University, The Netherlands
}
\date{}

\newtheorem{theorem}{Theorem}
\newtheorem{lemma}{Lemma}
\newtheorem{definition}{Definition}
\newtheorem{proposition}{Proposition}
\newtheorem{remark}{Remark}
\newtheorem{corollary}{Corollary}

\begin{document}

 \maketitle

\begin{abstract}
The Solow--Polasky diversity indicator (or magnitude) is a classical measure of diversity
based on pairwise distances. It has applications in ecology, conservation
planning, and, more recently, in algorithmic subset selection and
diversity optimization.
In this note, we investigate the computational complexity of selecting a subset of fixed cardinality from a finite set so as to maximize the Solow--Polasky diversity value.
We prove that this problem is NP-hard in general metric spaces.
The reduction is from the classical Independent Set problem and uses
a simple metric construction containing only two non-zero distance values.
\michael{Importantly, the hardness result holds for every fixed kernel parameter
$\theta_0>0$; equivalently, by rescaling the metric, one may fix the
parameter to $1$ without loss of generality.}
A central point is that this is not a boilerplate reduction: because
the Solow--Polasky objective is defined through matrix inversion, it is
a nontrivial nonlinear function of the distances.
Accordingly, the proof requires a dedicated strict-monotonicity argument
for the specific family of distance matrices arising in the reduction;
this strict monotonicity is established here for that family, but it is
not assumed to hold in full generality.
We also explain how the proof connects to continuity and monotonicity
considerations for diversity indicators.
\end{abstract}

\section{Introduction}

This paper proves that, for every fixed $\theta_0>0$, the problem of selecting a subset of cardinality $k$ from a given set $X$, with $k < |X|$, that maximizes the Solow--Polasky diversity is NP-hard in general metric spaces. This closes a
research gap made explicit in the recent comparative study of
Pereverdieva et al.~\cite{Pereverdieva2025}, where subset selection is
treated as a central theoretical property of diversity indicators: their
results establish NP-hardness for Riesz $s$-energy subset selection,
report NP-hardness for Max--min diversity, and leave the corresponding
entry for Solow--Polasky marked as open. More broadly, existing
complexity research on diversity has mainly focused either on diversity
objectives with a more direct metric or combinatorial structure, such as
diversity maximization in doubling metrics~\cite{Pellizzoni2023}, or on
the parameterized complexity of finding diverse collections of discrete
solutions~\cite{Baste2022}. By contrast, Solow--Polasky diversity is
defined through the inverse of a similarity matrix rather than through a
direct local function of the pairwise distances, so its hardness does
not follow in any immediate way from those earlier results.

Diversity measures based on pairwise distances play an important role
in ecology, conservation biology, and optimization.
One influential measure is the indicator proposed by
Solow and Polasky~\cite{SolowPolasky1994}, which evaluates the diversity
of a set through the inverse of a similarity matrix derived from an
exponential kernel. 

Let $S=\{x_1,\dots,x_k\}$ be a finite set equipped with pairwise distances
$d(x_i,x_j)$.
For a parameter $\theta>0$, the Solow--Polasky diversity is defined as

\[
SP_\theta(S)=\mathbf{1}^{\top} Z^{-1}\mathbf{1},
\]

where the similarity matrix $Z$ is

\[
Z_{ij}=e^{-\theta d(x_i,x_j)}.
\]

The vector $\mathbf{1}$ denotes the all-ones vector. {\andre Note that $SP_\theta(S)$ is equal to $\sum_j w_j$ where $w$ is the unique column vector satisfying $Z w = \mathbf{1}$ (in case $Z$ is invertible). }

Beyond ecological applications, diversity indicators have become relevant
for algorithmic subset selection and diversity optimization.
In this setting one is often given a finite set of candidate solutions
and asked to return a subset of prescribed cardinality that is as
diverse as possible according to a chosen indicator.
This perspective has been emphasized, among others, in the diversity
optimization work of Ulrich and coauthors~\cite{UlrichBaderThiele2010,Ulrich2012}
and more recently in the comparative study of Pereverdieva et al.~\cite{Pereverdieva2025}.
In particular, Ulrich explicitly formulates the problem of selecting,
from a given population, a subset of prescribed cardinality that maximizes
the Solow--Polasky diversity~\cite{Ulrich2012}. However, that work treats
the resulting optimization problem heuristically: exhaustive search over
all candidate subsets is noted to be infeasible due to combinatorial
explosion, and greedy deletion algorithms with polynomial running time
are proposed instead~\cite{UlrichBaderThiele2010,Ulrich2012}.

A closely related viewpoint comes from the theory of \emph{magnitude}.
\michael{Leinster~\cite{Leinster2021} provides a broad axiomatic treatment of
entropy and diversity, including similarity-based diversity measures and
their relation to magnitude. Also 
Huntsman~\cite{Huntsman2023} discusses the relationship between
Solow--Polasky diversity and the formally equivalent concept of magnitude
for exponential kernels, and explores its relevance in diversity
optimization.}
From that perspective, the subset selection problem studied here can
also be interpreted as the problem of maximizing the magnitude of a
chosen subset.

In this paper we study the following optimization problem.

\begin{definition}[Solow--Polasky subset selection problem]
\michael{Fix a constant $\theta_0>0$.
An instance consists of a finite metric space $(X,d)$ and an integer $k$
with $0\le k\le |X|$.
The task is to compute a subset
\[
S^* \in \arg\max_{S\subseteq X,\ |S|=k} SP_{\theta_0}(S).
\]
Equivalently, one may consider the associated decision problem asking,
for a given threshold $T$, whether there exists a subset $S\subseteq X$
with $|S|=k$ such that
\[
SP_{\theta_0}(S)\ge T.
\]}
\end{definition}

\begin{remark}[On fixing the kernel parameter]
\michael{Although the Solow--Polasky diversity is written with a parameter $\theta>0$,
this parameter does not need to be treated as part of the input from a
complexity-theoretic point of view.
Indeed, the similarity matrix is defined by
\[
Z_{ij}=e^{-\theta d(x_i,x_j)},
\]
so the objective depends on $\theta$ and the metric $d$ only through the
products $\theta d(x_i,x_j)$.
Hence any instance with parameter $\theta>0$ is equivalent, by the metric
rescaling
\[
d'(x_i,x_j)=\theta\, d(x_i,x_j),
\]
to an instance with fixed parameter $1$, because then
\[
e^{-d'(x_i,x_j)}=e^{-\theta d(x_i,x_j)}.
\]
Therefore one may fix the kernel parameter to any positive constant without
loss of generality.

In the proof below, however, we keep an arbitrary fixed constant $\theta_0>0$
explicit and choose the scale parameter $\lambda$ as a function of $\theta_0$.
This is only used to guarantee that the off-diagonal similarities
$q=e^{-\theta_0\lambda}$ and $r=e^{-2\theta_0\lambda}$ are sufficiently small
for the strict-monotonicity argument based on the Neumann-series bound.
Thus $\theta_0$ is not an input variable of the optimization problem; rather,
it is a fixed constant for which the reduction is calibrated.}
\end{remark}

Our main result is the following.

\begin{theorem}
For every fixed constant $\theta_0>0$, the problem of selecting, from a finite set $X$, a subset $S \subseteq X$ of cardinality $k<|X|$ that maximizes the Solow--Polasky diversity is NP-hard in general metric spaces.
\end{theorem}

\michael{
\begin{table}[h!]
\centering
\begin{tabular}{ll}
\toprule
Symbol & Meaning \\
\midrule
$X$ & finite ground set / metric space \\
$d$ & metric on $X$ \\
$S \subseteq X$ & selected subset \\
$k$ & prescribed subset cardinality \\
$\theta, \theta_0$ & kernel parameter; $\theta_0>0$ fixed in the problem definition \\
$SP_\theta(S)$ & Solow--Polasky diversity of $S$ \\
$Z$ & similarity matrix with entries $Z_{ij}=e^{-\theta d(x_i,x_j)}$ \\
$\mathbf{1}$ & all-ones vector \\
$I$ & identity matrix \\
$J=\mathbf{1}\mathbf{1}^{\top}$ & all-ones matrix \\
$G=(V,E)$ & graph instance of Independent Set \\
$\lambda$ & distance scale used in the reduction \\
$q$ & shorthand for $e^{-\theta_0\lambda}$ \\
$r$ & shorthand for $e^{-2\theta_0\lambda}=q^2$ \\
$F(t)$ & deformed objective value $\mathbf{1}^{\top}Z(t)^{-1}\mathbf{1}$ \\
$w(t)$ & vector $Z(t)^{-1}\mathbf{1}$ \\
$B(t)$ & off-diagonal perturbation in the decomposition $Z(t)=I+B(t)$ \\
$\|\cdot\|_\infty$ & maximum absolute row-sum norm \\
\bottomrule
\end{tabular}
\caption{Frequently used symbols.}
\end{table}
}

The reduction uses a simple metric with only two non-zero distances. However, unlike other diversity-based subset selection reductions, the distances are not just $0$, $1$, and $2$, but must be rescaled depending on $k$ and $\theta_0$.

This result also closes a natural open point left by the recent
comparative analysis of diversity indicators by Pereverdieva et al.~\cite{Pereverdieva2025}.
That work provided a broad theoretical comparison of different diversity
measures and cited or established NP-hardness results for them, while
the complexity status of Solow--Polasky subset selection remained open.
The present paper or note fills this gap for general metric spaces and thus can be seen as a sequel to the work by Pereverdieva et al..

\paragraph{Contributions.}
The contributions of this paper are as follows.

\begin{itemize}
\item We prove that maximizing Solow--Polasky diversity over subsets of
fixed size is NP-hard in general metric spaces.
\item The hardness result holds for every fixed kernel parameter
$\theta_0>0$.
\item The reduction uses a highly structured metric whose distance set is
contained in $\{0,\lambda,2\lambda\}$.
\item We provide a detailed proof that is easy to follow and makes the
strict-improvement argument explicit.
\item We clarify the role of monotonicity and continuity properties of
the diversity indicator and relate the result to total-distance,
max--min, and Riesz-energy based diversity optimization.
\end{itemize}

\section{Related Work and Context}

\subsection{Solow--Polasky and diversity axioms}

Solow and Polasky~\cite{SolowPolasky1994} introduced their diversity
measure as an axiomatic approach to biodiversity measurement based on
pairwise dissimilarities.
Their paper discusses properties such as monotonicity in varieties
(monotonicity under adding a genuinely new species), twinning, and
monotonicity in distance.
A key subtlety, however, is that the original paper explicitly states
that the Solow--Polasky indicator is \emph{not} monotone in distance
in general, and then conjectures that for the exponential kernel,
triangle inequality should be sufficient to restore monotonicity.

This issue is highlighted in Tamara Ulrich's work on diversity
optimization~\cite{UlrichBaderThiele2010,Ulrich2012}.
Ulrich adopts Solow--Polasky as one of the most important set-based
diversity measures for decision-space diversity, but also stresses that
the original proofs rely on assumptions that need not hold for arbitrary
pairwise distance matrices.
In particular, singular or nearly singular similarity matrices can cause
pathological values and monotonicity in distance is not known in full
generality for arbitrary numbers of points.

Huntsman~\cite{Huntsman2023} provides an additional theoretical angle
by relating Solow--Polasky diversity to magnitude.
This gives access to a broader conceptual framework for understanding the
indicator and suggests further applications in diversity optimization and
evolutionary search.
For the present work, the relevance of this connection is that the same
computational hardness result can be read both as a statement about
Solow--Polasky diversity and as a statement about magnitude-based subset
selection under exponential kernels.

Recent work by Pereverdieva et al.~\cite{Pereverdieva2025} places
these properties into a broader comparative framework.
That paper explicitly lists, among others, monotonicity in varieties,
twinning, monotonicity in distance, strict monotonicity in distance,
uniformity of optimal sets, computational effort, and subset selection
complexity as central theoretical properties of diversity indicators.
In their comparison, Solow--Polasky is listed as having monotonicity
in distance but not strict monotonicity in distance in the general
comparison table, while Riesz $s$-energy has both monotonicity and
strict monotonicity in distance, and max--min diversity has weak but
not strict monotonicity.

\subsection{Connection to other metric diversity objectives}

Our NP-hardness result fits naturally into the broader landscape of
metric diversity optimization.
The classical dispersion literature already shows that max--sum
(total distance, or sum of pairwise distances) and max--min objectives
are NP-hard even in specialized metric spaces~\cite{RaviRosenkrantzTayi1994}.

In the recent diversity-indicator literature, Pereverdieva et al.~\cite{Pereverdieva2025}
prove NP-hardness of Riesz $s$-energy subset selection in general metric spaces.
This is especially relevant because Riesz energy is another pairwise-distance
based indicator, but with a much more direct monotonicity structure:
increasing pairwise distances decreases energy and therefore strictly
improves the objective when the task is phrased as energy minimization.
By contrast, Solow--Polasky is defined via matrix inversion and is thus
more delicate analytically.

Taken together, these results suggest the following picture.
For several prominent diversity indicators in general metric spaces,
including total distance, max--min, Riesz energy, and now Solow--Polasky,
the associated cardinality-constrained subset selection problem is
computationally hard.
What differs across indicators is not so much the ultimate complexity
class, but rather the structural reasons behind the hardness and the
analytical tools needed to prove strict separation of good and bad
subsets.

\subsection{Why monotonicity and continuity matter here}

For max--sum diversity, strict improvement under increasing distances
is immediate from the objective definition.
For Riesz $s$-energy, strict monotonicity in distance is similarly
direct because enlarging distances strictly decreases the pairwise
interaction terms.
For max--min diversity, increasing distances cannot hurt, but strict
improvement may fail if the modified pair was not the bottleneck pair.

The Solow--Polasky case is more subtle.
The indicator depends on the full inverse of the similarity matrix,
so increasing one distance changes the objective in a nonlocal way.
This is precisely why the original high-level monotonicity principle
from Solow and Polasky is not, by itself, sufficient for the reduction
in this paper.
First, the original paper does not prove global monotonicity for all
metric instances; it explicitly notes counterexamples in general and
only conjectures that the exponential kernel with triangle inequality
should be safe.
Second, even weak monotonicity would not automatically give the
\emph{strict} separation required by the reduction:
we must show that every subset containing an edge has diversity
strictly smaller than an independent subset of the same cardinality.

For this reason, we prove a direct local monotonicity statement on the
family of matrices arising in our construction.
Continuity ensures that the objective varies smoothly under a one-parameter
deformation of one entry, and the derivative calculation shows that,
along this deformation, the objective increases strictly.
This is enough for the reduction and avoids relying on any unproved
global monotonicity claim.

\section{Continuity of the Diversity Indicator}

\michael{
\begin{proposition}[Continuity]
For every fixed finite set $S=\{x_1,\dots,x_k\}$, the map from the pairwise
distances to the Solow--Polasky diversity value,
\[
(d(x_i,x_j))_{1\le i,j\le k}
\longmapsto
SP_\theta(S),
\]
is continuous on the domain where the similarity matrix is nonsingular.
\end{proposition}
}

\begin{proof}
Each matrix entry
\[
Z_{ij}=e^{-\theta d(x_i,x_j)}
\]
depends continuously on the distances.
Matrix inversion is continuous on the set of nonsingular matrices.
Therefore the map
\[
Z \longmapsto \mathbf{1}^{\top} Z^{-1}\mathbf{1}
\]
is continuous wherever $Z$ is invertible.
Composing these maps yields the claim.
\end{proof}

\section{Reduction from Independent Set}
\label{sec:reduction}

\michael{We reduce from the classical NP-hard problem~\cite{GareyJohnson1979}.}

\begin{definition}[Independent Set]
Given a graph $G=(V,E)$ and an integer $k$, determine whether there
exists a subset $S\subseteq V$ with $|S|=k$ such that no two vertices
in $S$ are adjacent.
\end{definition}

\subsection{Encoding the graph as a metric}

Fix a constant $\theta_0>0$.
Given a graph $G=(V,E)$ we define a distance function on $V$ by

\[
d(u,v)=
\begin{cases}
0 & u=v,\\
\lambda & \{u,v\}\in E,\\
2\lambda & \{u,v\}\notin E,\ u\neq v,
\end{cases}
\]

where

\[
\lambda=\left\lceil \frac{\ln(4k)}{\theta_0}\right\rceil.
\]

\begin{lemma}
The function $d$ defines a metric.
\end{lemma}

\begin{proof}
The function is symmetric and vanishes exactly on the diagonal.
The only non-zero distances are $\lambda$ and $2\lambda$.
Since
\[
2\lambda \le \lambda+\lambda,
\]
{\andre and the non-negativity of $d$ } the triangle inequality holds for all triples of vertices.
\end{proof}

\subsection{Intuition}

Edges correspond to smaller distances while non-edges correspond to
larger distances.
Therefore subsets whose pairwise distances are all maximal are exactly
the independent sets of the graph.

\subsection{A small picture of the reduction}

Figure~\ref{fig:graph-metric} illustrates the transformation on a
four-vertex graph.

\begin{figure}[h]
\centering
\begin{tikzpicture}[node distance=1.5cm and 1.8cm, every node/.style={circle,draw,minimum size=7mm}]
\node (a1) at (0,1.2) {$1$};
\node (a2) at (1.5,2.0) {$2$};
\node (a3) at (1.5,0.4) {$3$};
\node (a4) at (3.0,1.2) {$4$};

\draw (a1) -- (a2);
\draw (a2) -- (a3);

\node[draw=none,rectangle] at (1.5,-0.6) {graph $G$: edges are $\{1,2\}$ and $\{2,3\}$};

\node[draw=none,rectangle] at (5.0,1.2) {$\Longrightarrow$};

\node[draw=none,rectangle,align=left] at (9.6,1.2) {metric encoding:\\[1mm]
$d(1,2)=d(2,3)=\lambda$,\\
$d(1,3)=d(1,4)=d(2,4)=d(3,4)=2\lambda$};

\end{tikzpicture}
\caption{Graph-to-metric transformation used in the reduction.}
\label{fig:graph-metric}
\end{figure}
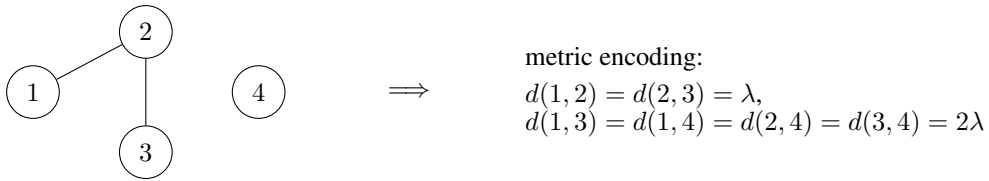

\subsection{Illustrative example}

Consider the graph with vertices $\{1,2,3,4\}$ and edges
\[
E=\{\{1,2\},\{2,3\}\}.
\]
Distances become
\[
d(1,2)=\lambda,\quad d(2,3)=\lambda,
\]
while
\[
d(1,3)=d(1,4)=d(2,4)=d(3,4)=2\lambda.
\]
The subset $\{1,3,4\}$ is independent, so all internal distances equal
$2\lambda$ and therefore produce maximal diversity among subsets of size
three. Any subset containing $\{1,2\}$ or $\{2,3\}$ has at least one
smaller distance $\lambda$ and therefore lower diversity.

\section{Similarity Parameters}

Define
\[
q=e^{-\theta_0\lambda},
\qquad
r=e^{-2\theta_0\lambda}.
\]

By construction,
\[
q\le \frac{1}{4k},
\qquad
r=q^2\le \frac{1}{16k^2}.
\]

\section{Diversity of Independent Sets}
\label{sec:divindset}
If $S$ is an independent set of size $k$, all pairwise distances in $S$
equal $2\lambda$.
\michael{
Hence the similarity matrix has the form

\[
Z=(1-r)I+rJ,
\]

where $I$ denotes the identity matrix, $\mathbf{1}$ denotes the all-ones
vector, and $J=\mathbf{1}\mathbf{1}^{\top}$ denotes the all-ones matrix.
}

A direct computation yields

\[
SP_{\theta_0}(S)=\frac{k}{1+(k-1)r}.
\]
{\andre This result easily emerges by solving the system of equations $Z w = \mathbf{1}$. Since in this case the system has precisely one solution. And the components of the unique column vector satisfying the equation are all equal to each other. Note also that the uniqueness of the weight vector $w$ proves the invertibility of $Z$ -- a fact which will be proved in the Appendix in more generality. }

\section{Strict Improvement When Increasing Distances}
\label{sec:analysis}

We now show, step by step, why a subset that contains an edge cannot be optimal.

\michael{Suppose a subset contains an edge $\{a,b\}$.}
Its corresponding similarity entry equals $q>r$.
The idea is to continuously increase the distance between $a$ and $b$
from $\lambda$ to $2\lambda$ and to track how the Solow--Polasky value changes.

Let $t\in[\lambda,2\lambda]$ and replace the distance $d(a,b)$ by $t$.
Let $Z(t)$ be the resulting similarity matrix and define

\[
F(t)=\mathbf{1}^{\top}Z(t)^{-1}\mathbf{1}.
\]

If we can show that $F(t)$ is strictly increasing in $t$, then replacing
an edge-distance $\lambda$ by the larger distance $2\lambda$ strictly
improves the objective.
This is exactly the kind of separation that the reduction needs.

\michael{The detailed derivative computation and the positivity argument are
deferred to Appendix~\ref{app:strict}. The key intermediate statement is
the following.}

\begin{lemma}
\michael{\label{lem:positive-w}}
For every $t\in[\lambda,2\lambda]$, all components of
$w(t)=Z(t)^{-1}\mathbf{1}$ are positive.
\end{lemma}

\michael{\begin{proof}
See Appendix~\ref{app:strict}.
\end{proof}}

\michael{Using the derivative identity established in Appendix~\ref{app:strict},
one obtains
\[
F'(t)=2\theta_0 e^{-\theta_0 t}\,w_a(t)w_b(t).
\]
By Lemma~\ref{lem:positive-w}, both factors $w_a(t)$ and $w_b(t)$ are positive
for all $t\in[\lambda,2\lambda]$, and therefore}
\[
F'(t)>0
\qquad\text{for all } t\in[\lambda,2\lambda].
\]
Thus increasing the distance from $\lambda$ to $2\lambda$ strictly
increases the diversity value.
Repeating this replacement for each internal edge of a subset shows
that every non-independent subset has strictly smaller diversity
than an independent subset of the same size.

\section{Discussion: Strict vs.\ Weak Monotonicity}

The monotonicity issue deserves an explicit discussion because it
separates Solow--Polasky diversity from more direct distance-based
indicators.

For max--sum diversity, strict improvement under increasing distances is
immediate from the objective definition.
For Riesz $s$-energy, strict monotonicity in distance is similarly
direct because enlarging distances strictly decreases the pairwise
interaction terms.
For max--min diversity, increasing distances cannot hurt, but strict
improvement may fail if the modified pair is not the bottleneck pair.

For Solow--Polasky diversity the objective depends on the inverse of the
full similarity matrix, so increasing a single distance changes the
objective in a nonlocal fashion.
This is why the original monotonicity discussion of Solow and Polasky is
not sufficient on its own for the present reduction.
What we need in the proof is not merely a weak monotonicity principle,
but a strict separation statement for the concrete family of instances
constructed from the graph.
The derivative argument and positivity lemma provide precisely this.

\section{Main Result}
\label{sec:result}

We now state the main hardness result. The proof is based on the reduction
from \textsc{Independent Set} developed in Section~\ref{sec:reduction},
with the distance scale chosen explicitly as
\[
\lambda=\left\lceil \frac{\ln(4k)}{\theta_0}\right\rceil,
\]
so that the corresponding similarity parameters satisfy
\[
q=e^{-\theta_0\lambda}\le \frac{1}{4k},
\qquad
r=e^{-2\theta_0\lambda}=q^2\le \frac{1}{16k^2}.
\]
These bounds are exactly what is needed in the strict-monotonicity
argument from Section~\ref{sec:analysis} and Appendix~\ref{app:strict}.

\begin{theorem}
For every fixed constant $\theta_0>0$, the following problem is NP-hard:
given a finite metric space $(X,d)$ and an integer $k$ with
$0 \le k \le |X|$, compute a subset
\[
S^* \in \arg\max_{S \subseteq X,\ |S|=k} SP_{\theta_0}(S).
\]
\end{theorem}

\begin{proof}[Proof sketch]
We reduce from \textsc{Independent Set}.
Given an instance $(G,k)$ with $G=(V,E)$, we construct the metric space
$(X,d)$ with $X=V$ as in Section~\ref{sec:reduction}, namely
\[
d(x,y)=
\begin{cases}
0, & x=y,\\
\lambda, & \{x,y\}\in E,\\
2\lambda, & \{x,y\}\notin E,\ x\neq y,
\end{cases}
\qquad
\text{where }
\lambda=\left\lceil \frac{\ln(4k)}{\theta_0}\right\rceil.
\]
By the lemma in Section~\ref{sec:reduction}, this is a metric, and the
construction is computable in polynomial time.

Now let $S\subseteq X$ with $|S|=k$.
Then $S$ is an independent set in $G$ if and only if every distinct pair
of points in $S$ has distance $2\lambda$.
In that case the corresponding similarity matrix has the special form
\[
Z=(1-r)I+rJ,
\]
and Section~\ref{sec:divindset} shows that such subsets
have Solow--Polasky value
\[
SP_{\theta_0}(S)=\frac{k}{1+(k-1)r}.
\]

If, on the other hand, $S$ contains an edge, then at least one pair in
$S$ has distance $\lambda$, so one off-diagonal similarity entry equals
$q$ instead of $r$.
By the strict-improvement argument from Section~\ref{sec:analysis},
whose derivative and positivity details are proved in
Appendix~\ref{app:strict}, replacing such a distance $\lambda$ by
$2\lambda$ strictly increases the Solow--Polasky value.
Hence every subset containing an edge has strictly smaller diversity than
an independent subset of the same cardinality.

Therefore $G$ has an independent set of size $k$ if and only if the
constructed Solow--Polasky instance has an optimal $k$-subset with all
pairwise distances equal to $2\lambda$.
Thus solving the Solow--Polasky subset selection problem would solve
\textsc{Independent Set}, and the problem is NP-hard.
\end{proof}

\begin{corollary}
For every fixed constant $\theta_0>0$, the Solow--Polasky subset
selection problem remains NP-hard when restricted to finite metric
spaces $(X,d)$ for which there exists a real number $\lambda>0$ such that
\[
d(x,y)\in\{0,\lambda,2\lambda\}
\qquad\text{for all }x,y\in X.
\]
\end{corollary}

\section{Conclusion}

This paper shows that, for every fixed $\theta_0>0$, maximizing
Solow--Polasky diversity over subsets of prescribed size is NP-hard in
general metric spaces.
This closes the complexity gap made explicit in the recent comparative
study of Pereverdieva et al.~\cite{Pereverdieva2025} and places
Solow--Polasky subset selection alongside total-distance/max--sum
diversity~\cite{RaviRosenkrantzTayi1994}, max--min diversity~\cite{Ghosh1996},
and Riesz-energy based subset selection~\cite{Pereverdieva2025} as
another computationally hard diversity optimization problem in general
metric spaces.

Technically, the proof is not a boilerplate reduction of the kind used
for more direct diversity objectives.
For max--sum and max--min, the effect of modifying one pairwise distance
is immediate from the objective definition.
For Solow--Polasky diversity, by contrast, the objective is defined
through the inverse of a similarity matrix, so a local change in one
distance propagates globally through the objective value.
Accordingly, weak monotonicity is not enough for the reduction.
What is needed is a strict separation between independent and
non-independent subsets, and this is established here only for the
bounded family of similarity matrices arising in the construction,
through a one-parameter deformation, a derivative formula for
$\mathbf{1}^{\top}Z^{-1}\mathbf{1}$, and a positivity argument for the
associated weight vector.

Several further routes suggest themselves from this result.
In particular, it would be natural to study the approximability of
Solow--Polasky subset selection, its parameterized complexity in the
sense of the broader diversity-of-solutions literature~\cite{Baste2022},
and its complexity on more structured metric classes such as Euclidean
spaces of fixed dimension or doubling metrics~\cite{Pellizzoni2023}.
These questions are compelling not only because they are standard next
steps after an NP-hardness result, but also because related complexity
research on diversity has already developed along precisely such lines
for structurally simpler objectives.

We have deliberately left these directions outside the scope of the
present paper in order not to overload it and to preserve a clear and
focused exposition of the main result and of the structural details of
its proof, in particular the strict-monotonicity argument needed in the
bounded regime arising in the reduction.
They are nonetheless well worth pursuing in future research.

\newpage

\appendix

\section{Appendix: Strict Improvement Argument}
\label{app:strict}

\michael{
This appendix contains the proof details used in
Section~\ref{sec:analysis}. In particular, we derive the formula for
$F'(t)$ and prove Lemma~\ref{lem:positive-w}.
}

\subsection{Derivative of the Solow--Polasky objective}

Let $t\in[\lambda,2\lambda]$ and let $Z(t)$ be the similarity matrix obtained
by replacing one internal edge-distance $d(a,b)$ by $t$.
Define
\[
F(t)=\mathbf{1}^{\top}Z(t)^{-1}\mathbf{1}.
\]

We differentiate the identity
\[
Z(t)^{-1}Z(t)=I.
\]
This gives
\[
\frac{d}{dt}Z(t)^{-1}
=
-\,Z(t)^{-1}Z'(t)Z(t)^{-1}.
\]

Therefore
\[
F'(t)
=
-\mathbf{1}^{\top}Z(t)^{-1}Z'(t)Z(t)^{-1}\mathbf{1}.
\]

Now define
\[
w(t)=Z(t)^{-1}\mathbf{1}.
\]
Then only the entries $(a,b)$ and $(b,a)$ of $Z(t)$ depend on $t$.
Both equal $e^{-\theta_0 t}$, so
\[
Z'_{ab}(t)=Z'_{ba}(t)=-\theta_0 e^{-\theta_0 t},
\]
and all other entries of $Z'(t)$ vanish.
Substituting this into the derivative formula yields
\[
F'(t)=2\theta_0 e^{-\theta_0 t}\,w_a(t)w_b(t).
\]

Thus $F'(t)>0$ will follow once we know that all components of $w(t)$ are
strictly positive.

\subsection{Positivity of \texorpdfstring{$w(t)$}{w(t)}}

\begin{lemma}
For every $t\in[\lambda,2\lambda]$, all components of
$w(t)=Z(t)^{-1}\mathbf{1}$ are positive.
\end{lemma}

\begin{proof}
Write
\[
Z(t)=I+B(t),
\]
where $B(t)$ has zero diagonal and off-diagonal entries in the interval
$[r,q]$.
Hence, for every row $i$,
\[
\sum_{j\ne i}|b_{ij}(t)|
=
\sum_{j\ne i} b_{ij}(t)
\le (k-1)q
\le \frac{k-1}{4k}
< \frac14.
\]
\michael{In particular},
\[
\|B(t)\|_\infty
=
\max_{1\le i\le k}\sum_{j=1}^k |b_{ij}(t)|
=
\max_{1\le i\le k}\sum_{j\ne i} |b_{ij}(t)|
<1.
\]
Therefore $I+B(t)$ is invertible by the standard row-diagonal-dominance
criterion. For completeness {\andre sake}, we state {\andre this as a lemmma with a proof }
in Appendix~\ref{app:invertibility}.

Since $\|B(t)\|_\infty<1$, the Neumann series converges and gives
\[
Z(t)^{-1}=(I+B(t))^{-1}=\sum_{m=0}^{\infty}(-B(t))^m.
\]
Applying this identity to the all-ones vector gives
\[
w(t)=\sum_{m=0}^{\infty}(-B(t))^m\mathbf{1}.
\]

For each component we obtain the lower bound
\[
w_i(t)\ge
1-\sum_{m=1}^{\infty}\|B(t)\|_\infty^m
>
1-\sum_{m=1}^{\infty}\left(\frac14\right)^m
=
1-\frac{1/4}{1-1/4}
=
\frac23>0.
\]
Thus every component of $w(t)$ is strictly positive.
\end{proof}

\section{Appendix: Matrix Invertibility Criterion}
\label{app:invertibility}

\michael{
The following elementary lemma is the clean formulation of the
invertibility fact used implicitly in the Neumann-series argument.
}

\begin{lemma}
\michael{\label{lem:IplusB-invertible}
Let $B=(b_{ij})\in\mathbb{R}^{N\times N}$ satisfy
\[
b_{ii}=0 \qquad \text{for all } i=1,\dots,N,
\]
and
\[
\sum_{j\ne i}|b_{ij}|<1
\qquad \text{for all } i=1,\dots,N.
\]
Then the matrix $I+B$ is invertible.}
\end{lemma}

\begin{proof}
\michael{
Let $M:=I+B$. Suppose, for contradiction, that $M$ is not invertible.
Then there exists a nonzero vector $x\in\mathbb{R}^N$ such that
\[
Mx=0.
\]
Choose an index $o\in\{1,\dots,N\}$ such that
\[
|x_o|=\max_{1\le i\le N}|x_i|.
\]
Since $x\neq 0$, we have $|x_o|>0$.

Now consider the $o$-th component of the equation $Mx=0$.
Since $m_{oo}=1+b_{oo}=1$ and $m_{oj}=b_{oj}$ for $j\ne o$, we obtain
\[
x_o+\sum_{j\ne o}b_{oj}x_j=0,
\]
and hence
\[
x_o=-\sum_{j\ne o}b_{oj}x_j.
\]
Taking absolute values and using the choice of $o$, we get
\[
|x_o|
\le
\sum_{j\ne o}|b_{oj}|\,|x_j|
\le
\left(\sum_{j\ne o}|b_{oj}|\right)|x_o|.
\]
Because $|x_o|>0$, division by $|x_o|$ yields
\[
1\le \sum_{j\ne o}|b_{oj}|,
\]
which contradicts the assumption
\[
\sum_{j\ne o}|b_{oj}|<1.
\]
Therefore $M=I+B$ is invertible.}
\end{proof}

\section{Appendix: Matrix Differentiation and the Neumann Series}

This appendix briefly summarizes two standard tools used in the proof of
strict improvement.
The goal is simply to make the argument more self-contained for readers
who may not work with matrix calculus every day.

\subsection{Differentiating the inverse of a matrix}

Let $Z(t)$ be a differentiable family of invertible matrices.
We start from the identity
\[
Z(t)^{-1} Z(t) = I.
\]
Differentiating both sides with respect to $t$ and using the product rule gives
\[
\frac{d}{dt}Z(t)^{-1}\, Z(t) + Z(t)^{-1} Z'(t) = 0.
\]
Multiplying on the right by $Z(t)^{-1}$ yields
\[
\frac{d}{dt}Z(t)^{-1} = - Z(t)^{-1} Z'(t) Z(t)^{-1}.
\]
This is the matrix-derivative identity used in the proof.

Now define
\[
F(t)=\mathbf{1}^{\top} Z(t)^{-1}\mathbf{1}.
\]
Then
\[
F'(t)
=
\mathbf{1}^{\top}\frac{d}{dt}Z(t)^{-1}\mathbf{1}
=
-\mathbf{1}^{\top} Z(t)^{-1} Z'(t) Z(t)^{-1}\mathbf{1}.
\]
If we write
\[
w(t)=Z(t)^{-1}\mathbf{1},
\]
this becomes
\[
F'(t) = -\, w(t)^{\top} Z'(t) w(t).
\]
So the derivative of the scalar objective is reduced to a quadratic form
in the derivative of the matrix.

In our setting, only the two symmetric entries $(a,b)$ and $(b,a)$
depend on $t$, and both are equal to $e^{-\theta_0 t}$.
Hence
\[
Z'_{ab}(t)=Z'_{ba}(t)=-\theta_0 e^{-\theta_0 t},
\]
and all other entries of $Z'(t)$ vanish.
Plugging this into the previous formula yields
\[
F'(t)=2\theta_0 e^{-\theta_0 t}\, w_a(t)w_b(t),
\]
which is the expression used in the proof.

\subsection{The Neumann series}

The Neumann series is the matrix analogue of the geometric series.
Recall the scalar identity
\[
\frac{1}{1+x}=1-x+x^2-x^3+\cdots
\qquad\text{for } |x|<1.
\]
For matrices, the corresponding statement is that if $\|B\|<1$ for a
submultiplicative matrix norm, then
\[
(I+B)^{-1} = \sum_{m=0}^{\infty} (-B)^m
= I - B + B^2 - B^3 + \cdots.
\]

A quick way to verify this is to look at the partial sums
\[
S_M=\sum_{m=0}^{M}(-B)^m.
\]
Then
\[
(I+B)S_M = I - (-B)^{M+1}.
\]
If $\|B\|<1$, then $(-B)^{M+1}\to 0$ as $M\to\infty$, so the partial sums
converge to the inverse:
\[
(I+B)\sum_{m=0}^{\infty}(-B)^m = I.
\]

In our proof we write
\[
Z(t)=I+B(t),
\]
where $B(t)$ has zero diagonal and small off-diagonal entries.
The bound
\[
\|B(t)\|_\infty < 1
\]
guarantees that the Neumann series converges:
\[
Z(t)^{-1} = (I+B(t))^{-1} = \sum_{m=0}^{\infty}(-B(t))^m.
\]
Applying this expansion to the all-ones vector gives
\[
Z(t)^{-1}\mathbf{1} = \sum_{m=0}^{\infty}(-B(t))^m\mathbf{1}.
\]
This is useful because the leading term is easy to understand, and the
remaining terms can be bounded by a geometric series.

\subsection{Remark on terminology}

The Neumann series is named after John von Neumann.
In the present paper, it is used only in a very basic way: as a convenient
tool to expand the inverse of a matrix whose off-diagonal part is small.

\section{Appendix: Numerical example}

As a small numerical sanity check of the reduction, we evaluated the
Solow--Polasky objective directly for one of the toy instances arising
from the construction.
This is only an illustration and not part of the proof.

Consider the graph on vertex set $\{1,2,3,4\}$ with edge set
\[
E=\{\{1,2\},\{2,3\}\},
\]
and let $k=3$.
Fix $\theta_0=1$.
Then the reduction chooses
\[
\lambda=\left\lceil \ln(4k)\right\rceil=\lceil \ln(12)\rceil=3,
\]
so that
\[
q=e^{-\lambda}=e^{-3}\approx 0.049787,
\qquad
r=e^{-2\lambda}=e^{-6}\approx 0.002479.
\]

The unique independent set of size $3$ is
\[
S_{\mathrm{ind}}=\{1,3,4\}.
\]
Its similarity matrix is
\[
Z_{\mathrm{ind}}=
\begin{pmatrix}
1 & r & r\\
r & 1 & r\\
r & r & 1
\end{pmatrix}
=(1-r)I+rJ.
\]
Hence
\[
SP_{\theta_0}(S_{\mathrm{ind}})
=
\frac{3}{1+2r}
\approx
\frac{3}{1+2e^{-6}}
\approx
2.985201.
\]

Now consider the non-independent subset
\[
S_{\mathrm{bad}}=\{1,2,4\},
\]
which contains the edge $\{1,2\}$.
Its similarity matrix is
\[
Z_{\mathrm{bad}}=
\begin{pmatrix}
1 & q & r\\
q & 1 & r\\
r & r & 1
\end{pmatrix}.
\]
By symmetry, if $w=Z_{\mathrm{bad}}^{-1}\mathbf{1}$ then
$w_1=w_2=a$ and $w_3=b$, where
\[
(1+q)a+rb=1,
\qquad
2ra+b=1.
\]
Solving gives
\[
a=\frac{1-r}{1+q-2r^2},
\qquad
b=\frac{1+q-2r}{1+q-2r^2},
\]
and therefore
\[
SP_{\theta_0}(S_{\mathrm{bad}})
=
2a+b
=
\frac{3+q-4r}{1+q-2r^2}
\approx
2.895737.
\]

Thus
\[
SP_{\theta_0}(S_{\mathrm{ind}})
>
SP_{\theta_0}(S_{\mathrm{bad}}),
\]
numerically confirming the behavior predicted by the reduction:
the independent subset has strictly larger Solow--Polasky value than a
subset of the same size containing an edge.

A short Python script reproducing this and a few similar checks is
available at\\ \url{https://github.com/emmerichmtm/SolowPolaskyReductionFromMaxIS}.

\label{tab:numerical-checks}

\end{document}